%% file: main.tex
\documentclass[conference]{IEEEtran}
\IEEEoverridecommandlockouts
\usepackage{amsmath,amssymb,mathtools,bm}
\usepackage{amsthm}
\usepackage{graphicx}
\usepackage{booktabs,tabularx,multirow}
\usepackage{algorithm}
\usepackage{algorithmic}
\usepackage{cite}
\usepackage{microtype}
\usepackage{xcolor}
\usepackage{url}
\usepackage{balance}
\usepackage{stfloats}
\input{numbers}

\newtheorem{theorem}{Theorem}
\newtheorem{proposition}{Proposition}
\newtheorem{corollary}{Corollary}

\newcommand{\cptp}{\textsc{cptp}}
\newcommand{\qmem}{\textsc{QMemTwin-AW}}
\newcommand{\qaic}{\textsc{QMemTwin-AIC}}
\newcommand{\onepm}{\textsc{1PM}}
\newcommand{\twopm}{\textsc{2PM}}
\newcommand{\passive}{\textsc{Passive-4}}
\newcommand{\ttmraw}{\textsc{TTM-Raw}}
\newcommand{\ttmcptp}{\textsc{TTM-CPTP}}
\newcommand{\mpencil}{\textsc{MatrixPencil}}
\newcommand{\ipinn}{\textsc{iPINN-1PM}}
\DeclareMathOperator{\AICc}{AICc}
\renewcommand{\arraystretch}{0.96}
\begin{document}

\title{Model-Uncertainty-Aware Digital Twins for
Non-Markovian Quantum Repeater Memories}
 \author{\IEEEauthorblockN{Muhammad Waqas Haseeb}
 \IEEEauthorblockA{}}
\maketitle

\begin{abstract}
Quantum-repeater controllers require reliable estimates of how stored
entanglement evolves while adjacent links are generated. A fixed
exponential lifetime is inexpensive but suppresses non-Markovian
revivals, whereas an overparameterized surrogate is unnecessary when
the correct microscopic family is known. We present QMemTwin, a
physics-constrained, model-uncertainty-aware digital twin that
calibrates a bank of contractive passive state-space models, combines
their amplitudes using finite-sample-corrected Akaike weights, and
abstains when temporal-holdout measurements indicate unsupported
dynamics. Passivity guarantees a completely positive and trace-
preserving amplitude-damping channel for every candidate and for their
convex amplitude average. We evaluate \Ncases{} synthetic memories,
including \Nsupported{} supported passive time-invariant cases and
\Nstructural{} nonstationary structural-shift cases. In the matched
one-pseudomode regime, the correctly specified \onepm{} fit remains
best. Across \Nunknown{} supported but family-uncertain cases, QMemTwin
reduces population RMSE by \PopReductionTwoPM{}\% and Bell-fidelity
RMSE by \BellReductionTwoPM{}\% relative to \twopm{}; both paired
effects remain significant after Holm correction. The always-flexible
\passive{} candidate achieves a slightly lower unknown-family mean but
is statistically indistinguishable from QMemTwin and incurs an
unnecessary accuracy penalty in the matched regime. The temporal-
holdout gate detects \StructuralDetectionPct{}\% of nonstationary
structural shifts at a \FalseAlarmPct{}\% false-alarm rate on supported
families, with AUROC \AdequacyAUC{}. In 2-, 4-, and 8-link repeater
chains, the mixed cohort shows no significant 8-link goodput gain,
whereas a predeclared high-backflow stress cohort yields a
\StrongEightGoodputGainPct{}\% gain over \twopm{} with paired
\(p=\StrongGoodputP\).
\end{abstract}


\section{Introduction}
A quantum Internet will distribute entanglement among remote processors, memories, and sensors~\cite{wehner2018quantum,azuma2023repeaters}. Multinode entanglement and teleportation have already been demonstrated in laboratory networks~\cite{pompili2021multinode,hermans2022teleportation}; networking research now addresses link protocols, routing, provisioning, fidelity constraints, and purification~\cite{dahlberg2019link,shi2020concurrent,zeng2022,yang2023,zhang2025,alkobaisi2025artificial}. These mechanisms consume a common primitive: an estimate of the state of a Bell pair that has waited in memory for a random time while other links were being established.

Most network studies receive a supplied Markovian lifetime or depolarizing channel. This abstraction is convenient, but a structured reservoir can return information to a memory, producing nonmonotone decay and disconnected future service windows. Detailed simulators such as NetSquid, SeQUeNCe, and QuNetSim can execute time-dependent physical models~\cite{coopmans2021netsquid,wu2021sequence,yu2007tornado}; they do not, by themselves, identify which memory model is supported by sparse characterization data.

When the correct microscopic model family is known, a calibrated physics-based description is generally preferable to a flexible neural surrogate because it is more interpretable, computationally efficient, and statistically economical. Quantum-repeater controllers, however, rarely possess such certainty. Sparse characterization measurements may be consistent with a single environmental pole, several independent poles, coherently coupled modes, a continuum-like reservoir, or dynamics that are not represented by the assumed model class. The central challenge is therefore not to replace established physics with a neural approximation, but to construct a reliable memory digital twin that remains useful under physical-model uncertainty.

This paper introduces QMemTwin, a model-uncertainty-aware digital twin for non-Markovian quantum-repeater memories. QMemTwin calibrates a bank of contractive passive state-space models of increasing complexity, combines their predictions using a finite-sample-corrected model-averaging criterion, propagates parameter and model uncertainty to Bell-pair observables and entanglement service windows, and applies a temporal-holdout adequacy gate to identify unsupported dynamics and abstain from unsafe predictions. The passive parameterization guarantees contractive evolution and valid amplitude-damping channels by construction, while sequential recalibration allows the twin to update as new characterization data become available. An inverse physics-informed neural network is included only as a comparative baseline rather than as the proposed method. This formulation positions QMemTwin as a cross-layer system-identification and uncertainty-quantification framework for network digital twins and quantum-repeater control.

Our contributions are:
\begin{itemize}
\item We formulate repeater-memory identification under three distinct regimes: matched known physics, supported family uncertainty, and structural model mismatch. The former ``OOD'' five-mode experiment is renamed \emph{higher-order mismatch}; a separate nonstationary process provides genuine structural shift.
\item We develop \qmem, an AICc-weighted bank of passive latent generators. A small-sample correction is necessary because the flexible candidate uses 18 parameters for $n=3N=72$ scalar observations.
\item We prove contractivity and \cptp{} validity of every candidate and their amplitude average, derive Bell-fidelity and Werner-swapping error bounds, and provide a margin certificate for threshold decisions.
\item We introduce a temporal-holdout adequacy score that combines normalized prediction error, candidate disagreement, and residual correlation. Unsupported dynamics trigger abstention instead of an overconfident service window.
\item We evaluate \Ncases{} memories, published baselines, uncertainty, sequential updates, and two network cohorts. We report paired bootstrap intervals, Wilcoxon tests, Holm correction, effect sizes, wins/ties/losses, and both favorable and unfavorable results.
\item We make transfer-tensor treatment auditable by reporting \ttmraw{} in every software benchmark and using a separately labeled \ttmcptp{} projection in the controller.
\end{itemize}

\begin{figure*}[t]
\centering
\includegraphics[width=0.98\textwidth]{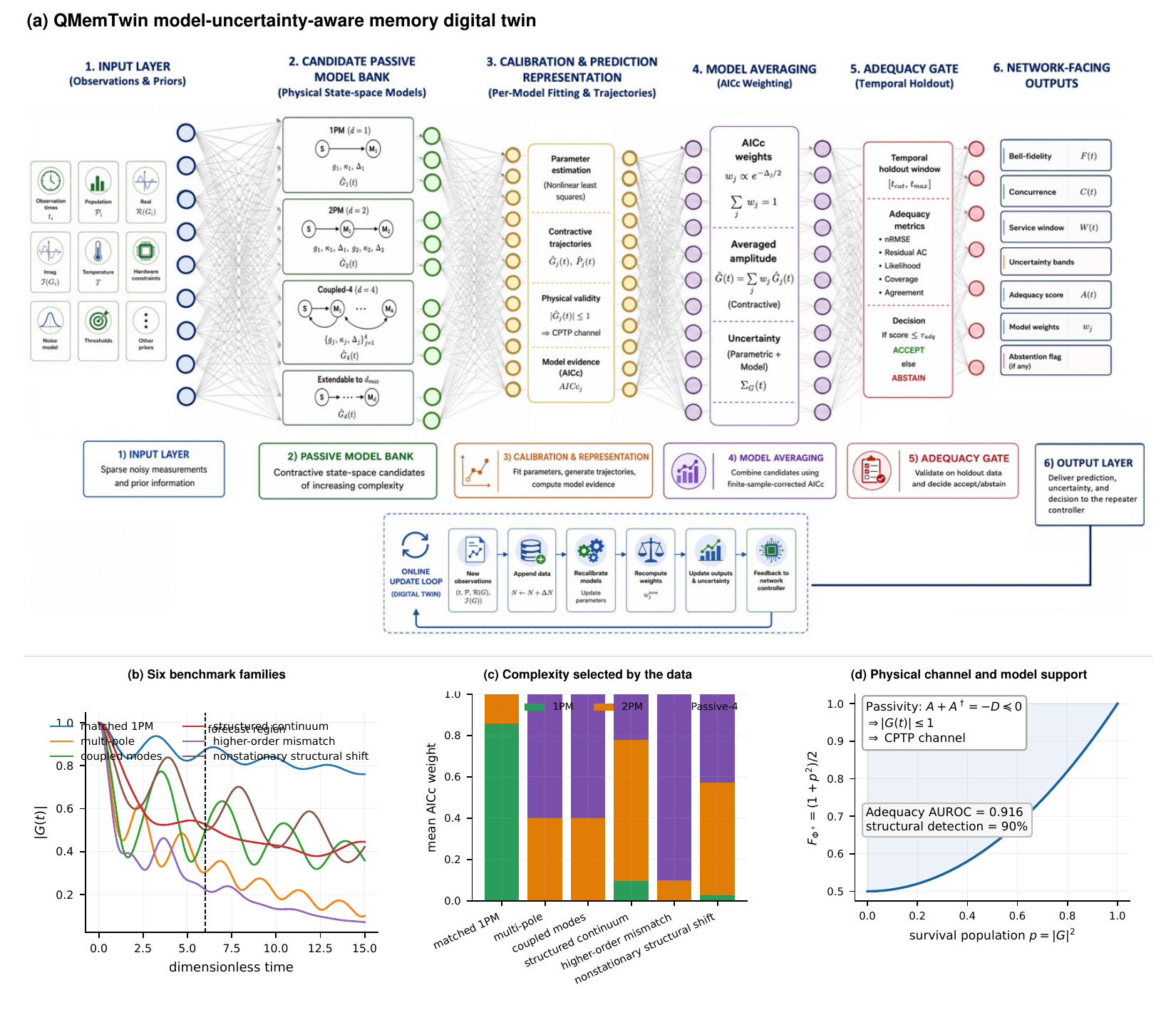}
\caption{QMemTwin workflow and evaluation regimes. Sparse quadrature/population measurements calibrate passive candidates. AICc weights adapt complexity, a temporal-holdout gate determines model support, and only supported physical trajectories are exposed as network-facing fidelity, concurrence, uncertainty, and service windows. The higher-order family tests capacity mismatch; the nonstationary family is the structural-shift test.}
\label{fig:framework}
\end{figure*}

\section{Related Work and Positioning}
\subsection{Quantum-network models and control}
Quantum-network control couples probabilistic optical generation, finite memories, classical heralding delays, local operations, and stochastic swapping. NetSquid supplies a discrete-event engine with extensible physical models~\cite{coopmans2021netsquid}; SeQUeNCe and QuNetSim provide complementary protocol abstractions~\cite{wu2021sequence,disiqueira2021qunetsim}. Routing and provisioning methods optimize entanglement throughput, virtual links, asynchronous generation, fidelity, and resource allocation~\cite{shi2020concurrent,zeng2022,yang2023,zhang2025}. QMemTwin is not another router. It supplies a calibrated memory-state interface to a controller and quantifies when that interface is unsupported.

\subsection{Non-Markovian System Identification under Model Uncertainty}
\label{sec:related_identification}

Non-Markovian dynamics admit several representations. Pseudomode
methods replace a structured reservoir with damped auxiliary modes,
embedding the reduced dynamics in an enlarged time-local model
\cite{garraway1997pseudomodes,pleasance2020pseudomodes}. For suitable
spectral densities, this construction is exact or systematically
improvable and motivates our passive state-space candidates. However,
it does not determine from sparse noisy data how many modes are needed,
whether coherent intermode couplings are required, or whether a
finite-dimensional time-invariant family is adequate.

Trajectory-based methods require fewer microscopic assumptions.
Transfer-tensor methods infer a discrete memory kernel from short-time
dynamical maps and extrapolate beyond the observation interval
\cite{cerrillo2014ttm}, but their standard form assumes
time-translation invariance and does not guarantee physical
contractivity under finite-data extrapolation. Process tensors provide
a general operational description of multi-time quantum processes and
interventions \cite{pollock2018process,alkobaisi2012interactive}, potentially with
tensor-network compression, but require richer multi-time
characterization data. Matrix-pencil methods efficiently identify
finite sums of damped exponentials \cite{hua1990matrix}, yet do not
independently enforce contractivity, complete positivity, or adequacy
of the selected exponential order.

Data-driven identification offers a complementary approach. Luchnikov
\emph{et al.} learn finite-dimensional Markovian embeddings from
sequential single-shot measurements \cite{luchnikov2020ml,ahmed2017open}, while
tensor-network models and physics-informed neural networks can learn
generic non-Markovian processes or solve forward and inverse quantum
master equations \cite{biswas2026,bae2007interactive}. Their physical validity,
extrapolation, computational cost, and interpretability depend on the
chosen parameterization and training procedure. When the correct
low-order microscopic family is already known, direct calibration is
the appropriate reference, as confirmed by our matched-family
experiments.

QMemTwin instead addresses uncertainty over which supported physical
family generated the sparse characterization data. It calibrates the
contractive passive model bank
\[
\mathcal{M}
=
\left\{
\mathrm{1PM},
\mathrm{2PM},
\mathrm{Passive\text{-}4}
\right\},
\]
combines candidate predictions using finite-sample-corrected Akaike
weights, and propagates parameter and model uncertainty to Bell-pair
observables and threshold-qualified service windows. Unlike hard model
selection, the average retains uncertainty when several candidates are
supported. Unlike unconstrained trajectory regressors, every candidate
and their convex amplitude average satisfy the contraction condition
required for a physical amplitude-damping channel.

QMemTwin does not assume that this bank represents every possible
non-Markovian process. A temporal-holdout adequacy gate evaluates
out-of-sample error, residual correlation, predictive coverage, and
candidate disagreement before releasing network-facing quantities.
The higher-order coupled benchmark tests latent-order mismatch within
the passive linear time-invariant class and is not treated as
structural out-of-distribution generalization. Structural mismatch is
tested separately with time-dependent passive generators outside the
fixed-generator bank, for which the intended response is abstention
rather than unsupported extrapolation.

For consistency, the transfer-tensor baseline is reported in two
forms. \emph{TTM-Raw} is the unconstrained extrapolation used in all
state- and fidelity-prediction benchmarks, together with its
channel-violation and divergence rates. Because an unphysical amplitude
cannot be supplied to the repeater controller, network experiments use
the separately labelled contractive projection \emph{TTM-CPTP}.
Table~\ref{tab:positioning} summarizes the methods by representation,
physical guarantees, adequacy testing, uncertainty treatment, and
network-level evaluation.

\begin{table}[t]
\centering
\caption{Method positioning. ``Supported'' refers to the method's
assumed family, not universal validity.}
\label{tab:positioning}
\scriptsize
\setlength{\tabcolsep}{2.2pt}
\renewcommand{\arraystretch}{0.95}
\begin{tabular}{@{}lcccc@{}}
\toprule
Method & Sparse & Order & CPTP & Net.\\
\midrule
Known-family fit & Yes & Fixed & By model & No\\
Transfer tensor & Short map & Order only & No$^\dagger$ & No\\
Process tensor & Multi-time & General & Conditional & No\\
Inverse PINN & Yes & Eq.-bound & Soft & No\\
Simulator & Supplied & N/A & Input & Yes\\
\qmem & Yes & AICc bank & Proven & Yes\\
\bottomrule
\end{tabular}

\vspace{1mm}
\raggedright\scriptsize
$^\dagger$Unless a separate physical projection is imposed.
\end{table}

\section{System Model and Identification Objective}
\subsection{Passive latent memory}
In the single-excitation sector, augment the observed memory coordinate with $m$ latent environmental coordinates:
\begin{equation}
\dot{\bm z}(t)=A\bm z(t),\quad \bm z(0)=\bm e_0,\quad
G(t)=\bm e_0^\dagger\bm z(t),
\label{eq:latent}
\end{equation}
where $G$ is the survival amplitude. We restrict each candidate generator to
\begin{equation}
A=-iH-\tfrac12D,\qquad H=H^\dagger,\quad D\succeq0.
\label{eq:passive}
\end{equation}
Independent pseudomodes use an arrowhead Hamiltonian,
\begin{align}
\dot c_s&=-i\sum_{r=1}^{m}g_ra_r,\nonumber\\
\dot a_r&=-\left(\kappa_r/2+i\omega_r\right)a_r-ig_rc_s,
\label{eq:star}
\end{align}
whereas the flexible four-mode candidate learns real symmetric mode--mode couplings as well. Positive couplings/damping rates are represented through bounded or softplus coordinates.

Tracing out the latent coordinates produces
\begin{equation}
\mathcal E_t(\rho)=K_0\rho K_0^\dagger+K_1\rho K_1^\dagger,
\quad
\begin{cases}
K_0=|0\rangle\!\langle0|+G(t)|1\rangle\!\langle1|,\\
K_1=\sqrt{1-|G(t)|^2}\,|0\rangle\!\langle1|.
\end{cases}
\label{eq:kraus}
\end{equation}
For two identical independent memories initially storing $|\Phi^+\rangle$, after local phase correction and with $p(t)=|G(t)|^2$,
\begin{equation}
F_{\Phi^+}(t)=\frac{1+p(t)^2}{2},\qquad C(t)=p(t)^2.
\label{eq:bell}
\end{equation}
An application's service set is
\begin{equation}
\mathcal W(F_{\min})=\{t\ge0:F_{\Phi^+}(t)\ge F_{\min}\},
\label{eq:window}
\end{equation}
which can contain disconnected intervals under information backflow.

\subsection{Measurements and three regimes}
At irregular times $t_i\le T_{\mathrm{obs}}$, the record contains noisy amplitude quadratures and population,
\begin{equation}
\widetilde G_i=G(t_i)+\epsilon_i^G,
\qquad \widetilde p_i=|G(t_i)|^2+\epsilon_i^p.
\label{eq:measurements}
\end{equation}
We distinguish:
\begin{enumerate}
\item \emph{matched}: a one-pseudomode truth and a \onepm{} candidate;
\item \emph{supported uncertainty}: multi-pole, coupled, continuum-discretized, and higher-order passive time-invariant dynamics;
\item \emph{structural shift}: a nonstationary channel-valid process outside the fixed-generator bank.
\end{enumerate}
The higher-order family is not called OOD: it is a five-mode coupled system evaluated with a four-mode flexible candidate, hence an order/capacity mismatch inside the same passive-LTI class.

For structural shift we use
\begin{equation}
G_{\mathrm{ns}}(t)=e^{-\Gamma(t)/2-i\phi(t)},\quad
\Gamma(t)=at+b[1-\cos(\Omega t)],
\label{eq:nonstationary}
\end{equation}
with $a,b>0$ and $\phi(t)=\omega_0t+c\sin(\Omega_\phi t)$. Since $\Gamma(t)\ge0$, $|G_{\mathrm{ns}}(t)|\le1$; however, the time-varying rate and infinite harmonic expansion are not exactly represented by the finite time-invariant candidate bank.

\begin{table}[t]
\centering
\caption{Principal notation.}
\label{tab:notation}
\footnotesize
\begin{tabularx}{\columnwidth}{@{}lX@{}}
\toprule
Symbol & Meaning\\
\midrule
$G(t),p(t)$ & survival amplitude and population\\
$A,H,D$ & passive generator, Hamiltonian, damping\\
$m,k_j$ & latent order and candidate parameter count\\
$N,n=3N$ & time samples and scalar residual count\\
$T_{\mathrm{obs}},T$ & observation and forecast horizons\\
$w_j$ & AICc weight of candidate $j$\\
$\mathcal W$ & threshold-qualified service set\\
$S_{\mathrm{adq}}$ & temporal-holdout adequacy score\\
$R_q,P_{\mathrm{viol}}$ & qualified goodput and violation rate\\
\bottomrule
\end{tabularx}
\end{table}

\section{QMemTwin Identification}
\subsection{Candidate calibration and AICc averaging}
The fused bank is $\mathcal M=\{\onepm,\twopm,\passive\}$ with $k_j\in\{3,6,18\}$. Markov, matrix-pencil, transfer-tensor, spline, and inverse-PINN methods are comparisons, not fused candidates. Each passive model minimizes the unpenalized observation residual plus a small numerical regularizer,
\begin{equation}
\mathrm{RSS}_j=\sum_{i=1}^{N}\!\left(
|\hat G_j(t_i)-\widetilde G_i|^2+
\beta[|\hat G_j(t_i)|^2-\widetilde p_i]^2\right),
\label{eq:rss}
\end{equation}
where the regularizer stabilizes optimization but is excluded from the information criterion.

With only $n=3N=72$ scalar residual components, the finite-sample correction is material for $k=18$. We use~\cite{akaike1974aic,hurvich1989aicc,burnham2004model}
\begin{align}
\mathrm{AIC}_j&=n\log(\mathrm{RSS}_j/n)+2k_j,\nonumber\\
\AICc_j&=\mathrm{AIC}_j+\frac{2k_j(k_j+1)}{n-k_j-1},
\label{eq:aicc}
\end{align}
then
\begin{align}
w_j&=\frac{\exp[-(\AICc_j-\AICc_{\min})/2]}
{\sum_{\ell\in\mathcal M}\exp[-(\AICc_\ell-\AICc_{\min})/2]},\nonumber\\
\hat G_{\mathrm{AW}}(t)&=\sum_jw_j\hat G_j(t).
\label{eq:average}
\end{align}
The uncorrected \qaic{} result is retained as a sensitivity ablation rather than silently selecting whichever criterion gives the lowest test error.

\subsection{Adequacy gate}
Model averaging is not evidence that the bank contains the truth. We reserve the latest quarter of characterization times as a temporal validation block, fit the bank to the earlier measurements, and compute
\begin{equation}
S_{\mathrm{adq}}=
\frac{\sqrt{E_G^2+0.7E_p^2}}{\sigma}
+0.35D_{\mathcal M}+0.50|\rho_1|,
\label{eq:adequacy}
\end{equation}
where $E_G$ and $E_p$ are holdout amplitude/population RMSE, $D_{\mathcal M}$ is normalized candidate disagreement, and $\rho_1$ is residual lag-one correlation. The threshold is fixed at the 95th percentile of 25 supported-family calibration cases. Evaluation uses a disjoint 25-case supported test set and all 10 structural-shift cases. A score above threshold causes abstention and requests a richer model or new measurements; it does not output an authoritative service window.

\subsection{Sequential updating and uncertainty}
A digital twin should update as telemetry arrives. For cumulative sample counts $N\in\{8,12,18,24\}$, candidates are refit and AICc weights recomputed. A parametric bootstrap perturbs the observed quadratures/populations at the measured noise level, recalibrates the bank, and forms pointwise population, fidelity, and service-window quantiles. This is computationally heavier than one deterministic fit but separates model uncertainty from a single optimizer trajectory.

\begin{algorithm}[t]
\caption{QMemTwin characterization and support gate}
\label{alg:qmem}
\footnotesize
\begin{algorithmic}[1]
\REQUIRE Measurements $\{t_i,\widetilde G_i,\widetilde p_i\}$, horizon $T$
\STATE Split chronologically into calibration and temporal holdout blocks
\FOR{$j\in\{\onepm,\twopm,\passive\}$}
\STATE Fit passive $A_j$ using~\eqref{eq:rss}; evaluate $\hat G_j(t)$
\ENDFOR
\STATE Compute AICc weights and $\hat G_{\mathrm{AW}}$ using~\eqref{eq:aicc}--\eqref{eq:average}
\STATE Evaluate $S_{\mathrm{adq}}$ by~\eqref{eq:adequacy}
\IF{$S_{\mathrm{adq}}$ exceeds calibrated threshold}
\RETURN \textsc{abstain}, residual diagnostics, request model expansion
\ENDIF
\STATE Bootstrap; map to $\hat F,\hat C,\hat{\mathcal W}$
\RETURN supported physical trajectory, weights, intervals, service windows
\end{algorithmic}
\end{algorithm}

\section{Formal Guarantees}
\begin{theorem}[Passive contraction and channel validity]
If $A=-iH-D/2$ with $H=H^\dagger$ and $D\succeq0$, then $\|e^{At}\bm e_0\|_2\le1$ and $|G(t)|\le1$ for every $t\ge0$. Consequently,~\eqref{eq:kraus} is \cptp.
\end{theorem}
\begin{proof}
Let $\bm z(t)=e^{At}\bm e_0$. Then
$\frac{d}{dt}\|\bm z\|_2^2=\bm z^\dagger(A+A^\dagger)\bm z=-\bm z^\dagger D\bm z\le0$.
Because $\|\bm z(0)\|_2=1$, contraction follows. Cauchy--Schwarz gives $|G|=|\bm e_0^\dagger\bm z|\le1$, and $K_0^\dagger K_0+K_1^\dagger K_1=I$.
\end{proof}

\begin{corollary}[Physicality of amplitude model averaging]
If $w_j\ge0$, $\sum_jw_j=1$, and $|\hat G_j(t)|\le1$, then $|\sum_jw_j\hat G_j(t)|\le1$; the QMemTwin average is therefore \cptp.
\end{corollary}
\begin{proof}
The triangle inequality yields $|\sum_jw_j\hat G_j|\le\sum_jw_j|\hat G_j|\le1$.
\end{proof}

\begin{proposition}[Observable and swap error]
Let $p=|G|^2$ and $\hat p=|\hat G|^2$. Then
\begin{equation}
|\hat F_{\Phi^+}-F_{\Phi^+}|\le|\hat p-p|\le2|\hat G-G|.
\label{eq:bellbound}
\end{equation}
For Werner swapping $S(F_1,F_2)=F_1F_2+(1-F_1)(1-F_2)/3$,
\begin{equation}
|S(\hat F_1,\hat F_2)-S(F_1,F_2)|
\le|\hat F_1-F_1|+|\hat F_2-F_2|.
\label{eq:swapbound}
\end{equation}
\end{proposition}
\begin{proof}
Equation~\eqref{eq:bell} gives $|\hat F-F|=|\hat p^2-p^2|/2\le|\hat p-p|$. The second inequality follows from $||\hat G|^2-|G|^2|\le(|\hat G|+|G|)|\hat G-G|$. Both partial derivatives of $S$ have magnitude at most one on $[1/4,1]^2$, so the mean-value theorem gives~\eqref{eq:swapbound}.
\end{proof}

\begin{proposition}[Threshold margin certificate]
If $|\hat F(t)-F(t)|<|F(t)-F_{\min}|$, then predicted and true threshold decisions agree.
\end{proposition}
\begin{proof}
A wrong binary decision requires the prediction error to cross the threshold, whose true margin is $|F-F_{\min}|$.
\end{proof}
The artifact verifies~\eqref{eq:swapbound} on \NswapSamples{} random draws with maximum observed ratio \SwapBoundMaxRatio{} and zero violations. Across all case/time/threshold triples, \MarginCertifiedPct{}\% are margin-certified, actual decision accuracy is \DecisionAccuracyPct{}\%, and certified-decision violations equal \DecisionBoundViolations{}.

\section{Experimental Methodology}
\subsection{Memory corpus}
We generate \Ncases{} deterministic-seed cases, \NperFamily{} per family, over $t\in[0,\Tend{}]$. Every case has \Nobs{} irregular noisy measurements restricted to $t\le\Tobs{}$; noise standard deviation is sampled from $[0.006,0.022]$. The six families are: one damped pseudomode; two--four independent modes; three--four coherently coupled modes; a 20-mode discretized Ohmic-plus-bands continuum; five coupled modes (higher-order mismatch); and the nonstationary process~\eqref{eq:nonstationary}. The first five are channel-valid passive-LTI cases. The exact reference for those cases is one eigendecomposition of the full generator; the structural family is evaluated analytically.

\subsection{Baselines and consistent TTM treatment}
All methods receive the same measurements. We compare Markov exponential calibration, \onepm, \twopm, \passive, \mpencil~\cite{hua1990matrix}, cubic spline, \ttmraw~\cite{cerrillo2014ttm}, and a representative \ipinn{} baseline following inverse open-system PINN practice~\cite{biswas2026}. The iPINN is run on one stratified case per family because its autograd calibration cost is substantially larger; it is a representative comparison, not a definitive architecture search.

Raw transfer-tensor propagation is evaluated on all software cases. It is not guaranteed contractive. For controller use only, we define the explicit radial projection
\begin{equation}
\widetilde G_{\mathrm{TTM}}(t)=
\frac{G_{\mathrm{TTM}}(t)}{\max\{1,|G_{\mathrm{TTM}}(t)|\}},
\label{eq:ttmprojection}
\end{equation}
and label it \ttmcptp{} everywhere. Raw and projected results are never merged under one name.

\subsection{Statistics and metrics}
Primary supported-family endpoints are population RMSE, Bell-fidelity RMSE, and mean service-window F1 at $F_{\min}\in\{0.65,0.72,0.80\}$. Secondary metrics include BLP backflow error~\cite{breuer2016colloquium}, validity, calibration time, uncertainty coverage, and threshold accuracy. Paired case effects use 5,000-resample bootstrap intervals and two-sided Wilcoxon signed-rank tests~\cite{wilcoxon1945}; the three primary $p$-values receive Holm correction. Family-wise $n=10$ effects are descriptive and are not advertised as separately significant.

\begin{figure*}[t]
\centering
\includegraphics[width=0.72\textwidth]{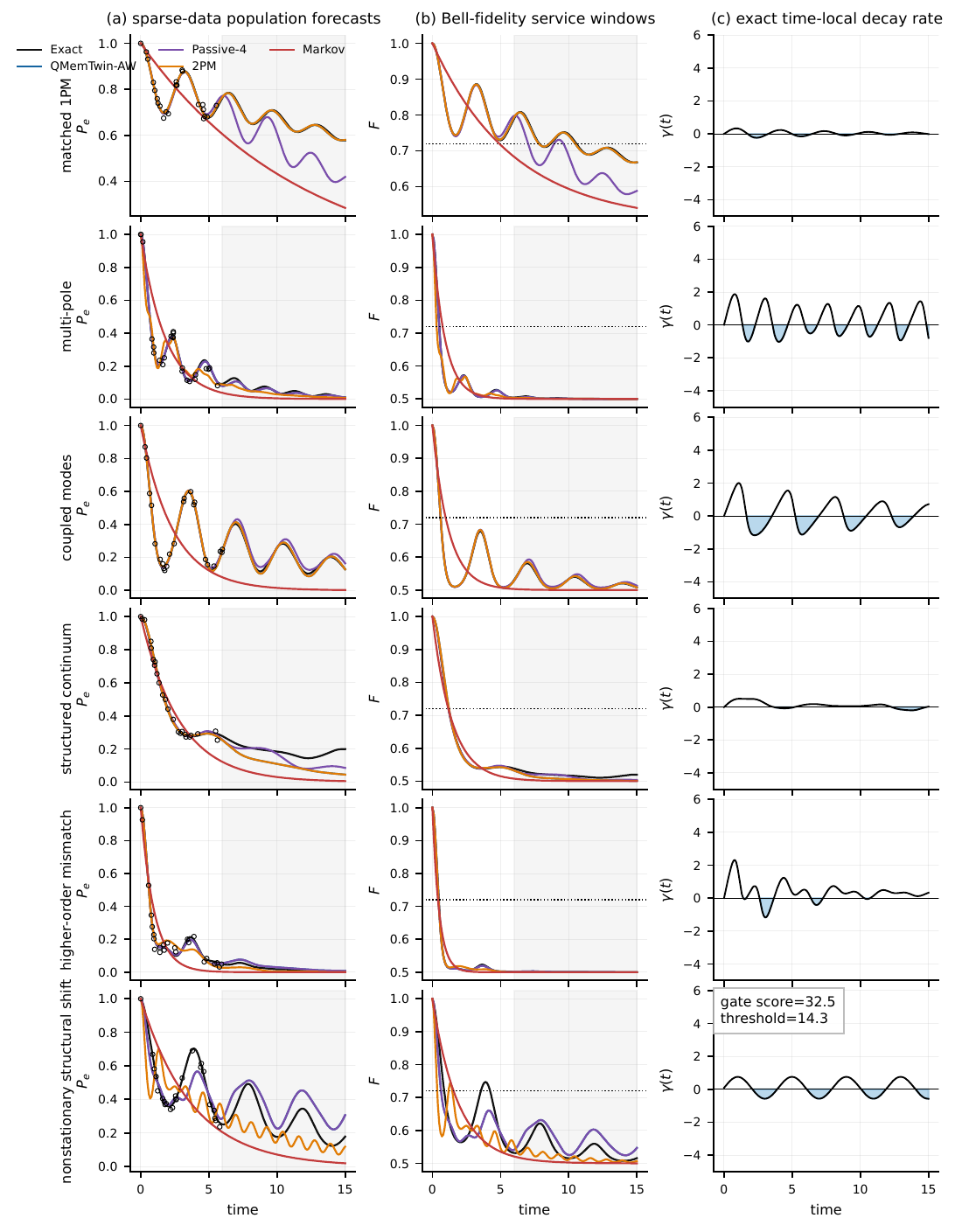}
\caption{Representative trajectories for all six families. Sparse points end at $T_{\mathrm{obs}}=6$; shading denotes forecasting. The final row is structurally nonstationary and is used to test the adequacy gate, not to claim universal extrapolation. Negative exact time-local rates identify information-backflow intervals.}
\label{fig:trajectories}
\end{figure*}

\begin{table*}[t]
\centering
\caption{Software benchmark. Supported-uncertain means multi, coupled, continuum, and higher-order families ($n=40$). Values are means except the raw TTM distribution, which is reported explicitly.}
\label{tab:software}
\footnotesize
\begin{tabular}{@{}lrrrrrr@{}}
\toprule
Method & Matched $p$ RMSE & Supported $p$ RMSE & Bell RMSE & Window F1 & Invalid cases & Median fit (s)\\
\midrule
\onepm & \MatchedOnePM & \UnknownOnePM & -- & -- & 0 & --\\
\twopm & -- & \UnknownTwoPM & \UnknownTwoPMBell & \UnknownTwoPMFone & 0 & \MedianTwoPMTime\\
\passive & \MatchedPassive & \UnknownPassive & -- & -- & 0 & \MedianPassiveTime\\
\qaic{} sensitivity & -- & \UnknownAIC & -- & -- & 0 & --\\
\qmem{} (AICc) & \MatchedQ & \UnknownQ & \UnknownQBell & \UnknownQFone & 0 & \MedianQTime\\
Markov & -- & \UnknownMarkov & -- & -- & 0 & --\\
\ttmcptp & -- & \UnknownTTMCPTP & -- & -- & 0 & \MedianTTMTime\\
\ttmraw & -- & median \TTMRawMedian & same$^\dagger$ & same$^\dagger$ & \TTMRawInvalid/\TTMN & \MedianTTMTime\\
\bottomrule
\end{tabular}\\[1mm]
\parbox{0.98\textwidth}{\scriptsize \ttmraw{} population RMSE has mean \TTMRawMean, IQR $[\TTMRawQOne,\TTMRawQThree]$, and maximum \TTMRawMax. $^\dagger$Bell/service metrics clip probabilities internally and therefore correspond to the projected trajectory; they are not used to conceal raw population divergence.}
\end{table*}

\section{Identification Results}
\subsection{Matched and supported-family regimes}
Table~\ref{tab:software} reports the result that motivated the redesign. In the matched regime, \onepm{} population RMSE is \MatchedOnePM, lower than QMemTwin's \MatchedQ and far below always-flexible \passive{} (\MatchedPassive). When the microscopic family is known, the known-family fit should be used.

Across \Nunknown{} supported but family-uncertain cases, QMemTwin obtains population RMSE \UnknownQ versus \UnknownTwoPM for \twopm, a \PopReductionTwoPM{}\% reduction. The paired mean reduction is \PopEffect (95\% CI $[\PopCILow,\PopCIHigh]$), with \PopWins/40 wins and Holm-adjusted $p=\PopPHolm$. Bell RMSE falls from \UnknownTwoPMBell to \UnknownQBell; the paired reduction is \BellEffect (95\% CI $[\BellCILow,\BellCIHigh]$, adjusted $p=\BellPHolm$). Service-window F1 increases by \WindowEffect (95\% CI $[\WindowCILow,\WindowCIHigh]$, adjusted $p=\WindowPHolm$), comprising \WindowWins{} improvements, \WindowTies{} ties, and no losses.

The flexible \passive{} candidate has a slightly lower unknown-family mean than QMemTwin by $\lvert\PassiveEffect\rvert$, but the paired difference is not significant ($p=\PassiveP$). This is important: QMemTwin's value is not that averaging dominates the flexible model in every unknown case. It preserves low-order matched performance while remaining statistically comparable to the flexible candidate under uncertainty, and it exposes candidate weights rather than pretending the high-order model is always necessary.
Uncorrected \qaic{} has lower unknown-family mean than AICc by $\lvert\AICEffect\rvert$ ($p=\AICP$). We retain AICc as the primary estimator because the correction is mandated by the ratio of $k$ to $n$; Fig.~\ref{fig:distribution} makes this sensitivity visible instead of selecting AIC after inspecting test error.

\begin{figure*}[t]
\centering
\includegraphics[width=0.93\textwidth]{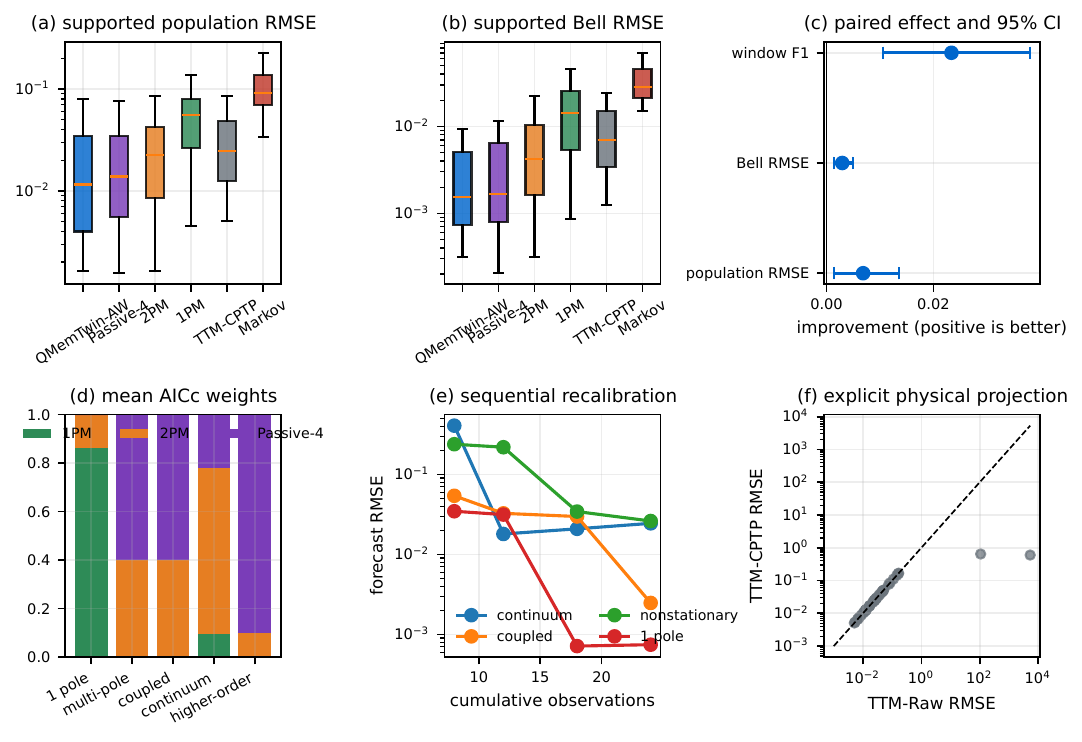}
\caption{Prediction and adaptive-complexity audit: supported-family errors, paired 95\% intervals, AICc weights, sequential forecast error, and raw/projected TTM. Detailed six-panel source figures remain in the artifact.}
\label{fig:distribution}
\end{figure*}

\begin{table}[t]
\centering
\caption{Primary paired effects: QMemTwin-AW minus/better-than 2PM over 40 supported-family cases.}
\label{tab:stats}
\scriptsize
\begin{tabular}{@{}lrrrr@{}}
\toprule
Endpoint & Mean effect & 95\% CI & Holm $p$ & W/T/L\\
\midrule
$p$ RMSE reduction & \PopEffect & [\PopCILow,\PopCIHigh] & \PopPHolm & \PopWins/0/9\\
Bell RMSE reduction & \BellEffect & [\BellCILow,\BellCIHigh] & \BellPHolm & \BellWins/0/8\\
Window-F1 increase & \WindowEffect & [\WindowCILow,\WindowCIHigh] & \WindowPHolm & \WindowWins/\WindowTies/0\\
\bottomrule
\end{tabular}
\end{table}

\subsection{Adaptive complexity and sequential updates}
AICc weights provide an interpretable complexity trace. Single-pole cases favor \onepm; coupled and higher-order cases activate \passive; continuum cases alternate between \twopm{} and \passive{} according to resolvable structure. Fig.~\ref{fig:distribution} also shows why finite-sample correction matters: AICc and AIC agree on easy cases but can differ when a small RSS decrease is purchased with twelve additional parameters.

Sequential recalibration is operational rather than rhetorical. In representative single and coupled memories, forecast RMSE decreases from \SequentialSingleStart{} to \SequentialSingleEnd{} and from \SequentialCoupledStart{} to \SequentialCoupledEnd{} as measurements accumulate from 8 to 24. Complexity weights can change discontinuously when later observations resolve a revival; that is expected model evidence, not training instability. The structural case may obtain a seemingly confident internal weight while still failing the independent adequacy gate, demonstrating why model selection and model support are different questions.

\subsection{Structural mismatch and abstention}
The adequacy threshold \AdequacyThreshold{} is calibrated without structural cases. On the disjoint test set, the gate detects \StructuralDetectionPct{}\% of nonstationary shifts at a \FalseAlarmPct{}\% false-alarm rate and AUROC \AdequacyAUC{}. Figure~\ref{fig:adequacy} decomposes the score. Structural cases have higher temporal holdout error and disagreement, but no single component is sufficient in every case. One structural case is missed, and one supported case is rejected. These failures are reported because an adequacy gate is a risk-control mechanism, not proof of universal OOD detection.

The correct operational response to a rejected case is abstention, collection of later-time measurements, or expansion to a nonstationary/finite-temperature model. We do not report the structural-family RMSE as evidence of successful extrapolation, even though some passive approximations happen to fit individual cases.

\begin{figure*}[t]
\centering
\includegraphics[width=0.93\textwidth]{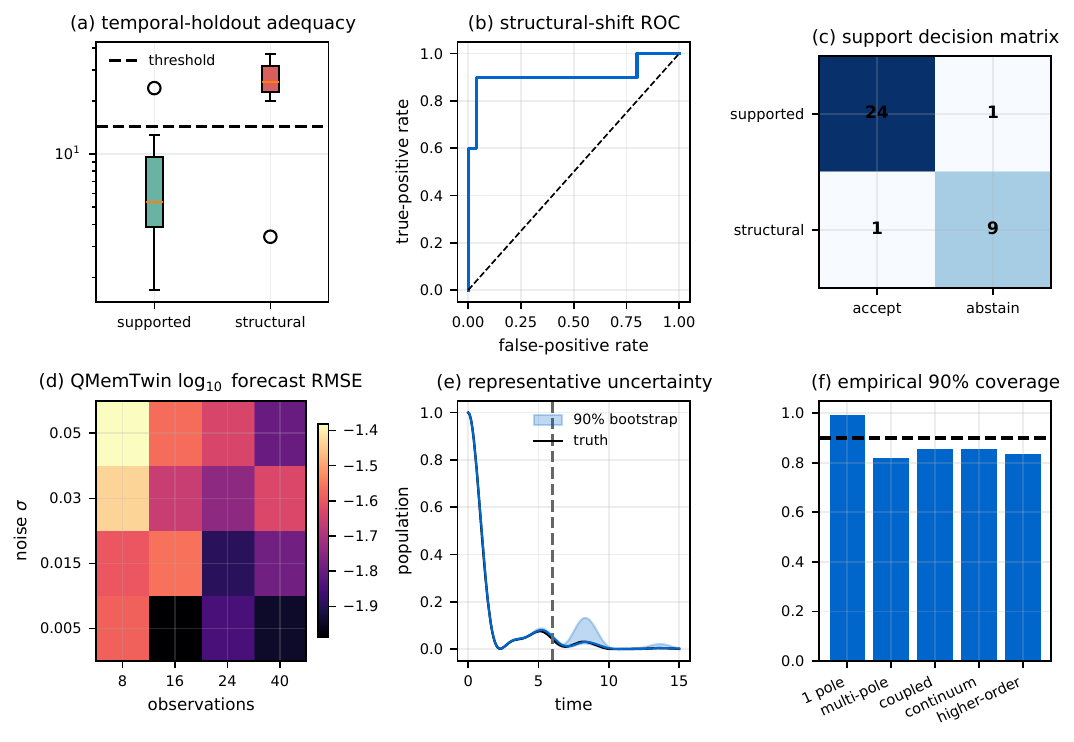}
\caption{Adequacy, robustness, and uncertainty: support-score separation, structural-shift ROC/confusion, observation/noise sensitivity, bootstrap trajectory, and empirical coverage.}
\label{fig:adequacy}
\end{figure*}

\subsection{Robustness, uncertainty, and formal checks}
Across the observation-count/noise sweep, QMemTwin mean forecast RMSE is \RobustQ versus \RobustTwoPM for \twopm. The gain is heterogeneous: high noise and eight observations can leave all candidates poorly identified, while later measurements often expose coupled structure. Nominal 90\% bootstrap intervals achieve mean population/Bell coverage \BootstrapCoveragePct{}\%; the undercoverage is visible and precludes describing the intervals as perfectly calibrated.

The lower half of Fig.~\ref{fig:adequacy} shows representative uncertainty bands, robustness surfaces, and the Bell-error certificate. QMemTwin's maximum predicted amplitude is \MaxPredAmplitude{} up to floating-point roundoff; no channel-domain violation occurs. These guarantees apply within the amplitude-damping representation and do not validate the model family itself---the role of the adequacy gate.

\subsection{Published baselines and cost}
Figure~\ref{fig:baselines} reports calibration cost versus accuracy. \ttmraw{} is fast but exhibits catastrophic extrapolation in \TTMRawInvalid{} of \TTMN{} supported cases: its mean population RMSE is \TTMRawMean, median \TTMRawMedian, and maximum \TTMRawMax. The radial projection removes channel-domain violations and reduces mean error to \TTMCPTPMean, but it is a post-hoc feasibility repair, not a guarantee that the recovered memory kernel is correct.

Matrix pencil is also fast and competitive in several multi-exponential cases, but its pole truncation and stability projection degrade continuum/coupled cases. The representative inverse-PINN subset demonstrates that a neural baseline can fit some trajectories but does not supply the passivity guarantee or model-adequacy diagnosis. Its limited six-case comparison is not used for a broad superiority claim.

QMemTwin's median calibration time \MedianQTime{}~s is higher than \twopm{} (\MedianTwoPMTime{}~s) because all three candidates are calibrated, and slightly above \passive{} (\MedianPassiveTime{}~s). This cost occurs during characterization; table lookup/interpolation during network control is negligible. The paper therefore claims improved safety and adaptation, not lower offline fitting cost.

\begin{figure*}[t]
\centering
\includegraphics[width=0.93\textwidth]{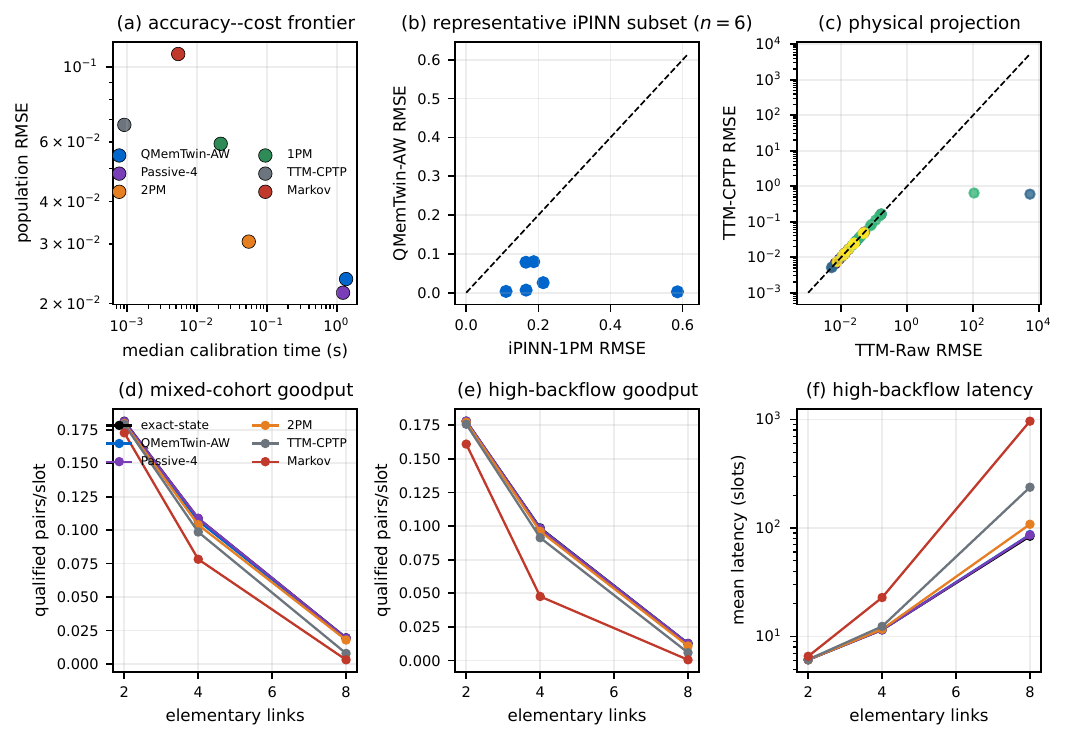}
\caption{Baseline audit and network consequences: accuracy--cost, representative inverse-PINN comparison, raw/projected TTM, and mixed/high-backflow repeater cohorts over 12 paired seeds.}
\label{fig:baselines}
\end{figure*}

\section{Multi-Hop Repeater Consequences}
\subsection{Controlled slotted benchmark}
We use a transparent slotted line network with $L\in\{2,4,8\}$ elementary links, 3,000 slots, and 12 paired seeds per method/cohort. Each link attempts generation with probability $0.28$, initial fidelity $0.97$, buffer capacity four, and BSM success probability $0.88$. The application threshold is $F_{\min}=0.68$. Pairs age according to heterogeneous true traces; the controller sees only its predicted traces. It retains a below-threshold pair when a future service window lies inside the horizon, otherwise discards it; compatible pairs can be purified with a Werner-state BBPSSW abstraction~\cite{bennett1996purification,briegel1998repeater}, and adjacent equal-length segments are greedily swapped.

All methods share event seeds and policy code; only memory prediction changes. The exact-state input is called a heuristic, not an oracle upper bound, because a greedy policy with different earlier decisions changes future stochastic opportunities. The simulator is intentionally smaller than NetSquid and is used for controlled attribution, not hardware timing fidelity.

Two cohorts are declared. The \emph{mixed supported} cohort samples all four supported uncertain families. The \emph{high-backflow stress} cohort restricts link traces to BLP measure above 0.15. The second is not a replacement for the first; it tests the regime in which memory-model differences should matter most.

\begin{algorithm}[t]
\caption{Memory-aware repeater policy}
\label{alg:network}
\footnotesize
\begin{algorithmic}[1]
\FOR{each slot}
\STATE Attempt elementary generation in non-full buffers
\STATE Age all segments using true/predicted memory traces
\STATE Discard only if no predicted service window lies in horizon $H$
\STATE Purify compatible same-segment pairs when predicted gain is positive
\FOR{swap levels $\ell=0,\ldots,\log_2L-1$}
\STATE Swap adjacent equal-length segments with qualified predicted output
\ENDFOR
\STATE Deliver predicted-qualified end-to-end pairs
\STATE Log true qualification, violation, latency, occupancy, and operations
\ENDFOR
\end{algorithmic}
\end{algorithm}

\subsection{Network results}
In the mixed 8-link cohort, QMemTwin qualified goodput is \AllEightQGoodput{} pairs/slot versus \AllEightTwoPMGoodput{} for \twopm, a \AllEightGoodputGainPct{}\% mean gain. The paired confidence interval crosses zero and the Wilcoxon test is not significant; this is reported in the lower half of Fig.~\ref{fig:baselines}. Mean latency is lower by \AllEightLatencyGain{} slots, with borderline paired evidence. Thus the broad cohort supports a modest trend, not a universal network improvement.

In the predeclared high-backflow 8-link stress cohort, QMemTwin reaches \StrongEightQGoodput{} versus \StrongEightTwoPMGoodput{}, a \StrongEightGoodputGainPct{}\% gain. The paired effect is \StrongGoodputEffect{} pairs/slot (95\% CI $[\StrongGoodputCILow,\StrongGoodputCIHigh]$, $p=\StrongGoodputP$). Mean latency improves by \StrongLatencyEffect{} slots (95\% CI $[\StrongLatencyCILow,\StrongLatencyCIHigh]$, $p=\StrongLatencyP$). Violation-rate differences are not significant ($p=\StrongViolationP$), so we do not claim a reliability reduction in this revised experiment.

The always-flexible \passive{} controller is competitive in both cohorts, consistent with its physical RMSE. \ttmcptp{} is conservative and slow at eight links, while raw TTM is never supplied to the controller. Markov performance degrades sharply with depth because it suppresses future windows. These results establish that memory identification can affect network behavior.

\section{Conclusion}
\label{sec:conclusion}

This paper introduced QMemTwin, a model-uncertainty-aware digital twin for identifying and predicting non-Markovian quantum-repeater memory dynamics from sparse characterization measurements. QMemTwin calibrates a bank of contractive passive state-space models, combines supported candidates using finite-sample-corrected Akaike weights, propagates parameter and model uncertainty to network-facing quantities, and applies a temporal-holdout adequacy test before releasing predictions. The passive parameterization guarantees \(|G(t)|\leq 1\), and the convex amplitude average preserves this property, ensuring completely positive and trace-preserving amplitude-damping channels. When the correct physical family is known, direct calibration remains preferable: in the matched one-pseudomode regime, \onepm{} attains a population RMSE of \(0.005460\), compared with \(0.005895\) for QMemTwin. Across 40 family-uncertain memories, however, QMemTwin reduces population RMSE relative to \twopm{} by \(0.006854\) with a 95\% confidence interval of \([0.001464,0.013523]\) and Holm-corrected \(p=0.00510\). Bell-fidelity RMSE decreases by \(0.002946\), with a 95\% confidence interval of \([0.001361,0.004905]\) and \(p=2.07\times10^{-5}\), while service-window F1 increases by \(0.023387\), with a 95\% confidence interval of \([0.010620,0.038113]\), no observed degradations, and \(p=0.00294\). Passive-4 achieves a slightly lower mean population error, but its difference from QMemTwin is not statistically significant; the model average instead retains simpler models when supported and exposes uncertainty when several descriptions remain plausible.

QMemTwin does not assume that its candidate bank contains every possible memory process. The higher-order coupled benchmark represents latent-order mismatch within the passive linear time-invariant class, while structural mismatch is evaluated separately using time-dependent passive generators outside the fixed-generator bank. The temporal-holdout adequacy gate detects 90\% of these structural-shift cases, produces a 4\% false-alarm rate on supported memories, and attains an AUROC of \(0.916\); rejected cases trigger abstention rather than unsupported long-horizon prediction. The multi-hop simulations further show that memory-model accuracy can affect network operation, although the benefit depends on the regime. In the mixed supported-memory cohort, the eight-link qualified-goodput difference relative to \twopm{} is not statistically significant. In the predeclared high-backflow stress cohort, QMemTwin increases eight-link fidelity-qualified goodput by 15.2\% and reduces mean delivery latency by \(21.34\) slots, demonstrating the importance of model uncertainty when delayed revivals and nonmonotonic dynamics influence retain, discard, purification, and swapping decisions.

The conclusions remain limited to zero-temperature,
single-excitation, local amplitude-damping dynamics. Finite-temperature
excitation, leakage, correlated memories, active control pulses,
detector errors, fabrication variability, and long-term experimental
drift are not included. The AICc calculation uses a weighted-residual
likelihood approximation with \(n=3N\), and alternative likelihood
models, Bayesian evidence, or hierarchical priors may alter the
candidate weights. Finally, the network layer uses Werner-state
abstractions and a greedy control policy and is not intended to replace
NetSquid, optical-repeater experiments, or full hardware-level network
simulation. Overall, the central contribution is the separation among
matched physics, supported model uncertainty, and detected structural
mismatch: known low-order models should be used directly, contractive
model averaging should be used when several physical descriptions
remain plausible, and the twin should abstain when temporal validation
indicates inadequate model support.

\balance
\begingroup
\def\IEEEbibitemsep{0pt}
\renewcommand{\footnotesize}{\fontsize{7.8}{8.3}\selectfont}
\bibliographystyle{IEEEtran}
\bibliography{references}
\endgroup
\end{document}

%% file: numbers.tex
\newcommand{\Ncases}{60}
\newcommand{\Nsupported}{50}
\newcommand{\Nunknown}{40}
\newcommand{\Nstructural}{10}
\newcommand{\NperFamily}{10}
\newcommand{\Nobs}{24}
\newcommand{\Tobs}{6}
\newcommand{\Tend}{15}
\newcommand{\MatchedOnePM}{0.00546}
\newcommand{\MatchedQ}{0.00589}
\newcommand{\MatchedPassive}{0.06507}
\newcommand{\UnknownQ}{0.02364}
\newcommand{\UnknownAIC}{0.02181}
\newcommand{\UnknownPassive}{0.02154}
\newcommand{\UnknownTwoPM}{0.03049}
\newcommand{\UnknownOnePM}{0.05928}
\newcommand{\UnknownMarkov}{0.10912}
\newcommand{\UnknownTTMCPTP}{0.06742}
\newcommand{\UnknownQBell}{0.00599}
\newcommand{\UnknownTwoPMBell}{0.00894}
\newcommand{\UnknownQFone}{0.982}
\newcommand{\UnknownTwoPMFone}{0.959}
\newcommand{\PopReductionTwoPM}{22.5}
\newcommand{\BellReductionTwoPM}{33.0}
\newcommand{\PopEffect}{0.00685}
\newcommand{\PopCILow}{0.00146}
\newcommand{\PopCIHigh}{0.01352}
\newcommand{\PopPHolm}{0.0051}
\newcommand{\PopWins}{31}
\newcommand{\BellEffect}{0.00295}
\newcommand{\BellCILow}{0.00136}
\newcommand{\BellCIHigh}{0.00491}
\newcommand{\BellPHolm}{2.07e-05}
\newcommand{\BellWins}{32}
\newcommand{\WindowEffect}{0.0234}
\newcommand{\WindowCILow}{0.0106}
\newcommand{\WindowCIHigh}{0.0381}
\newcommand{\WindowPHolm}{0.00294}
\newcommand{\WindowWins}{13}
\newcommand{\WindowTies}{27}
\newcommand{\PassiveEffect}{-0.00209}
\newcommand{\PassiveP}{0.927}
\newcommand{\AICEffect}{-0.00183}
\newcommand{\AICP}{0.040}
\newcommand{\AdequacyAUC}{0.916}
\newcommand{\StructuralDetectionPct}{90}
\newcommand{\FalseAlarmPct}{4}
\newcommand{\AdequacyThreshold}{14.29}
\newcommand{\TTMRawMean}{135.945}
\newcommand{\TTMRawMedian}{0.02479}
\newcommand{\TTMRawQOne}{0.01253}
\newcommand{\TTMRawQThree}{0.04816}
\newcommand{\TTMRawMax}{5329.0}
\newcommand{\TTMRawInvalid}{2}
\newcommand{\TTMN}{40}
\newcommand{\TTMCPTPMean}{0.06742}
\newcommand{\RobustQ}{0.0220}
\newcommand{\RobustTwoPM}{0.0303}
\newcommand{\BootstrapCoveragePct}{87.0}
\newcommand{\MedianQTime}{1.356}
\newcommand{\MedianTwoPMTime}{0.055}
\newcommand{\MedianPassiveTime}{1.227}
\newcommand{\MedianTTMTime}{0.0009}

\newcommand{\MaxPredAmplitude}{1.000000}
\newcommand{\MarginCertifiedPct}{96.8}
\newcommand{\DecisionAccuracyPct}{98.5}
\newcommand{\DecisionBoundViolations}{0}
\newcommand{\NswapSamples}{120000}
\newcommand{\SwapBoundMaxRatio}{0.99996}
\newcommand{\AllEightQGoodput}{0.01922}
\newcommand{\AllEightTwoPMGoodput}{0.01792}
\newcommand{\AllEightGoodputGainPct}{7.3}
\newcommand{\AllEightLatencyGain}{16.77}
\newcommand{\StrongEightQGoodput}{0.01244}
\newcommand{\StrongEightTwoPMGoodput}{0.01081}
\newcommand{\StrongEightGoodputGainPct}{15.2}
\newcommand{\StrongGoodputEffect}{0.00164}
\newcommand{\StrongGoodputCILow}{0.00042}
\newcommand{\StrongGoodputCIHigh}{0.00278}
\newcommand{\StrongGoodputP}{0.033}
\newcommand{\StrongLatencyEffect}{21.34}
\newcommand{\StrongLatencyCILow}{9.57}
\newcommand{\StrongLatencyCIHigh}{33.23}
\newcommand{\StrongLatencyP}{0.009}
\newcommand{\StrongViolationP}{0.938}
\newcommand{\SequentialSingleStart}{0.0346}
\newcommand{\SequentialSingleEnd}{0.00075}
\newcommand{\SequentialCoupledStart}{0.0540}
\newcommand{\SequentialCoupledEnd}{0.00248}